\documentclass{article}
\usepackage[utf8]{inputenc}
\usepackage{amsmath, amsthm, amssymb, amsfonts}
\usepackage{bm}
\usepackage[authoryear]{natbib}
\usepackage{dsfont}
\usepackage{graphicx}
\usepackage{subcaption}
\usepackage[dvipsnames]{xcolor}
\usepackage[linesnumbered, ruled,vlined]{algorithm2e}
\usepackage[a4paper, total={6in, 8in}]{geometry}
\usepackage{appendix}
\usepackage{hyperref}
\usepackage{booktabs}

\newtheorem{proposition}{Proposition}%[section]
\newtheorem{assumption}{Assumption}%[section]
\newtheorem{lemma}{Lemma}%[section]

\usepackage{authblk}
\title{Spatiotemporal blocking of the bouncy particle sampler for efficient inference in state space models}
\author[1]{Jacob Vorstrup Goldman}
\author[1]{Sumeetpal S. Singh}
\affil[1]{Signal Processing and Communications Laboratory, Department of Engineering, University of Cambridge}
\date{}

\begin{document}

\maketitle

\begin{abstract}
We propose a novel blocked version of the continuous-time bouncy particle sampler of \citep{bouchard2018bouncy} which is applicable to any differentiable probability density.  This alternative implementation is motivated by blocked Gibbs sampling for state space models \citep{singh2017blocking} and leads to significant improvement in terms of effective sample size per second, and furthermore, allows for significant parallelization of the resulting algorithm. The new algorithms are particularly efficient for latent state inference in high-dimensional state space models, where blocking in both space and time is necessary to avoid degeneracy of MCMC. The efficiency of our blocked bouncy particle sampler, in comparison with both the standard implementation of the bouncy particle sampler and the particle Gibbs algorithm of \cite{andrieu2010particle}, is illustrated numerically for both simulated data and a challenging real-world financial dataset. 
\end{abstract}

\section{Introduction}
\subsection{Background}
Markovian state space models are a class of probabilistic graphical models applied in biology, signal processing, target tracking, finance and more, see \cite{ryden2006inference} for a technical overview. In our setup, a latent process $(x_n, n \geq 1)$ on $\mathbb{R}^d$ evolves according to a state transition density $p(x_n | x_{n-1})$, with $p(\cdot)$ denoting a generic density. The dimension of the latent process is its spatial dimension, although often no physically relevant interpretation might be available. We indirectly observe the latent process through a noisy set of observations $(y_n, N \geq n \geq 1)$ defined on $\mathbb{R}^m$, where the realizations depend only on the current value of the latent state, $y_n | x_n \sim p(y_n| x_n)$. For convenience we introduce the sequence notation $i:j=(i,i+1,\ldots,j)$ when $j>i$. Unless otherwise mentioned, the sequence is $y_{1:N}$ fixed throughout. Given the observation sequence, we define smoothing as the off-line estimation of the conditional joint probability density $p(x_{l:m} | y_{1:N})$, with $1 \leq l \leq m \leq N$. We will be interested in the case where the target is the full conditional $p(x_{1:N} | y_{1:N})$. Smoothing is generally a hard problem due to the high dimensionality of the state space and spatiotemporal interdependence of the latent states; below we will give a brief historical overview, and subsequently detail our contributions. \\
Sequential Monte Carlo methods form the backbone of most smoothing algorithms. A popular early example is the sequential importance resampling smoother of \cite{kitagawa1996monte}, which utilises the entire trajectories and final particle weights of a particle filter to generate smoothed estimates. This method suffers from severe particle degeneracy as the resampling step non-strictly decreases the available paths used to estimate the joint posterior. A solution was the algorithm of \cite{godsill2004monte}, which introduces a sequence of backward passes incorporating the state transition. This algorithm has linear computation cost in time, particles and number of samples. Similar algorithms like the general two-filter smoother of \cite{briers2010smoothing} have equivalent computational costs. In \cite{finke2017approximate}, an approximate localization scheme is proposed for the forward-backward algorithm, including theoretical results that guarantees bounds on the asymptotic variance and bias under models that are sufficiently local. In the landmark paper of \cite{andrieu2010particle}, the authors introduced particle Markov Chain Monte Carlo, which combines particle filters in conjunction with either Metropolis-Hastings or Gibbs algorithms. The latter algorithm, denoted particle Gibbs, generates a single trajectory chosen according to the final particle weights from a particle filter run conditionally on a fixed trajectory. Particle Gibbs is stable if the number of particles grow at least linearly with the time series length; further theoretical analysis of ergodicity and asymptotic variance is provided in \cite{andrieu2018uniform} and \cite{chopin2015particle}. More recently, couplings of conditional particle filters have been introduced in \cite{jacob2020smoothing, lee2020coupled}, and provide unbiased estimators with asymptotically exact confidence intervals. \\

Unfortunately, the performance of particle Gibbs depends entirely on the efficiency of the conditional particle filter which like the particle filter can suffer from weight degeneracy. If the spatial dimension is large, the curse of dimensionality described in \cite{bengtsson2008curse} implies that infeasibly many particles are required to effectively approximate the posterior; localization of proposals by exploiting spatial conditional independence was subsequently introduced in \cite{rebeschini2015can} but this method is not generically applicable. As an alternative, the space-time particle filter \cite{beskos2017stable} is applicable if the likelihood can be written in a product form of terms that depend on an increasing number of latent dimensions. In the data assimilation field, a very popular method for high-dimensional filtering is the use of the Ensemble Kalman Filter algorithm, but the theoretical understanding of this algorithm is still quite limited, see however \cite{del2018stability, de2018long} for recent work in this regard.
Overall, there is no generically applicable, asymptotically exact approach that makes the particle filter viable in high dimensional time-series models.

In comparison with filtering which is known to be uniformly stable in time under reasonable assumptions, see \cite{van2019particle}, the difficulty of smoothing increases as the length of the time-series increases. In such scenarios, \cite{whiteley2010bs}, in the Royal Statistical Society's discussion of \cite{andrieu2010particle}, proposed to incorporate a backward pass similar to the algorithm of \cite{godsill2004monte} to avoid particle paucity in the early trajectories; for low spatial dimensions, the resulting algorithm was shown in \cite{lee2020coupled} to be computationally efficient and stable as the time horizon grows. A conceptually similar method that updates the fixed reference trajectory has been developed in \cite{lindsten2014particle}. As an alternative to manipulation of particle lineages, applying the particle Gibbs algorithm inside a generic Gibbs sampler over temporal blocks is proposed in \cite{singh2017blocking}, where the authors furthermore show a stable mixing rate as the length of the time series increases. \cite{singh2017blocking} also shows that the sharing of states via overlapping blocks increases the mixing rate as the overlap increase. While the issue of long time series has been addressed successfully by the algorithms detailed above, the curse of spatial dimensionality indicates that particle Gibbs and more sophisticated extensions are currently unworkable in practical smoothing applications featuring high spatial dimensions. 
\subsection{Contributions}
As a solution to the issues in high dimension, we propose a novel blocked sampling scheme based on irreversible, continuous-time piecewise deterministic Markov processes. Methods based on this class of stochastic process were originally introduced as event-chain Monte Carlo in the statistical physics literature by \cite{bernard2009event}, and subsequently further developed in the computational statistics literature recently, see for example \cite{bouchard2018bouncy, bierkens2019zig, wu2020coordinate, power2019accelerated}. In practice, the algorithms iterate persistent dynamics of the state variable with jumps to its direction at random event-times. They also only depend on evaluations of the gradient of the log-posterior. Local versions of these samplers, see \cite{bouchard2018bouncy} and \cite{bierkens2020piecewise}, can exploit any additive structure of the log-posterior density to more efficiently update trajectories, however as discussed above, long range dependencies of states indicate that sharing of information is desirable to achieve efficient mixing. To allow for sharing of information, we introduce a blocked version of the bouncy particle sampler of \cite{bouchard2018bouncy} that utilises arbitrarily designed overlapping blocks. {(Our resulting algorithm is different from the approach presented in \cite{zhao2019analysis} where the BPS is run on conditional distributions in a Metropolis-within-Gibbs type fashion.)} The blocking scheme is implementable without any additional assumptions on the target distribution, and is therefore useful for generic target densities, particularly in cases where the associated factor graph is highly dense. 

As our second contribution, we introduce an alternative implementation of the blocked sampler that leverages partitions to simultaneously update entire sets of blocks. The number of competing exponential clocks in the resulting sampler is independent of dimension and thus feature $O(1)$ clocks for any target, and allows, for the first time for a piecewise-deterministic Monte Carlo algorithm, to carry out parallel updates at event-times. Our numerical examples indicate that the blocked samplers can achieve noteworthy improvements compared to the bouncy particle sampler, both in terms of mixing time and effective sample size per unit of time, even without the use of parallelization. In addition, the blocked sampler provides efficient sampling of state space models when particle Gibbs methods, which are widely considered state of the art for state space models, fail due to high spatial dimensions. 
\section{Setup}
\subsection{Notation}
In what follows, subscript on a variable $x$ will denote temporal indices, while superscript indicates spatial. By $x \sim \mathcal{N}(0, 1)$ we mean that $x$ is distributed as a standard normal variable, whereas we by $\mathcal{N}(x; 0, 1)$ mean the evaluation at $x$ of the standard normal density; this notation is extended to other densities. We use the standard sequence notation $i{:}j = (i, i+1, \ldots, j-1, j)$ and $[n] = (1,2,\ldots, n-1, n)$.   A generic Poisson process is denoted by $\Pi$ and the associated, possibly inhomogeneous, rate function is the function $t \mapsto \lambda(t)$. Let $M_{m, n}$ be the space of $m \times n$ real-valued matrices, with $m$ referring to row and $n$ to columns, respectively. We denote by $\star$ the Hadamard product operator. The standard Frobenius norm of a matrix $X \in M_{m,n}$ is denoted $\Vert X \Vert_F = \sqrt{\text{tr}(X^T X)} = \sqrt{\sum_i \sum_j x_{i,j}^2}$, and the Frobenius inner product with another matrix $Y \in M_{m,n}$ is subsequently $\langle X, Y \rangle_F = \text{tr}(X^T Y) = \sum_i \sum_j x_{i,j}y_{i,j}$. 
\subsection{State space models} 
The class of state space models we consider have differentiable transition and observation densities
\begin{alignat*}{3}
&p(x_1) &&= \exp \Big \{ -f_1(x_1)\Big \}, &&f_1 \in C^1(\mathbb{R}^d \rightarrow \mathbb{R}), \\
&p(x_n | x_{n-1}) &&= \exp\Big \{-f(x_{n-1}, x_n) \Big \},\quad &&f \in C^1(\mathbb{R}^d \times \mathbb{R}^d \rightarrow \mathbb{R}), \\ 
&p(y_n | x_n) &&= \exp\Big\{-g(x_n, y_n)\Big\}, &&g\in C^1(\mathbb{R}^d \times \mathbb{R}^m \rightarrow \mathbb{R}).
\end{alignat*}
It is thus natural to work in log-space for the remainder of the paper, and we note in this regard that all probability density functions are assumed to be normalised. The exponential notation is therefore merely a notational convenience to avoid repeated mentions of log-densities. We also only require access to derivatives of $f$ and $g$ which may have more convenient expressions than the full probability distribution.  The negative log of the joint state density of the entire latent state $x \in M_{d,N}$ is denoted the potential energy $U: M_{d,N} \rightarrow \mathbb R$, and is given as
\begin{eqnarray*}
U(x_{1:N}) \equiv -\log \pi(x_{1:N} \mid y_{1:N}) =  f_1(x_1) + g(x_1,y_1) + \sum_{n=2}^N f(x_{n-1}, x_n) + g(x_n, y_n).
%\frac{p(x_{1:N}, y_{1:N})}{p(y_{1:N})} = \frac{1}{p(y_{1:N})}p(x_1) p(y_1|x_1) \prod_{n=2}^N p(x_n|x_{n-1})p(y_n|x_n), 
%\\ \begin{\propto \exp\Big \{-\frac{1}{2}\Big (\sum_{t=1}^T (x_t - A x_{t-1})^T \Sigma^{-1} (x_t - A x_{t-1})- 2 g(y_t|x_t) \Big )  \Big \} 
\end{eqnarray*}
To ease notation we have dropped the explicit dependence on $y_{1:N}$ when writing the log conditional joint state density from now on. We will often need to refer to the derivative $\partial U / \partial x$, which we denote as the matrix map $\nabla U: M_{d,N} \rightarrow M_{d,N}$ where the entry in the $k$'th row and $n$'th column is given by the partial derivative $\nabla U(x)_{k,n} = \partial U(x)/\partial x_n^k.$ Again, we remind the reader that subscript on a variable $x$ will denote temporal indices, while superscript indicates spatial. 

\subsection{Blocking strategies}\label{seq:block_strat}
Recall that $[n] = (1,2,\ldots, n-1, n)$. A \emph{blocking strategy} $\overline B$ is a cover of the index of set of the latent states $I = [d] \times [N]$, and solely consists of rectangular subsets. A generic block $B$ is always of the form $i{:}j \times l{:}m$ with $i <j, l<m$, with the coordinates referring to spatial and temporal dimensions, respectively. The size of a block is the ordered pair $(|i{:}j|, |l{:}m|)$. Blocks are allowed to overlap, we denote by the interior of a block the indices that it does not share with any other block. The neighbourhood set of a block is %the blocking substrategy
\[
N(B) = \{ B' \in \overline{B} \mid B \cap B' \neq \emptyset \},
\]
and always includes the block itself.
A blocking strategy is temporal if each block in a strategy is of the form $1{:}d \times l{:}m$, these are the most natural strategy types to consider for state space models and will be the general focus in the rest of the paper, but the methods presented below work for arbitrary strategies.
\begin{figure}
	\begin{center}
	\includegraphics[width=1\textwidth]{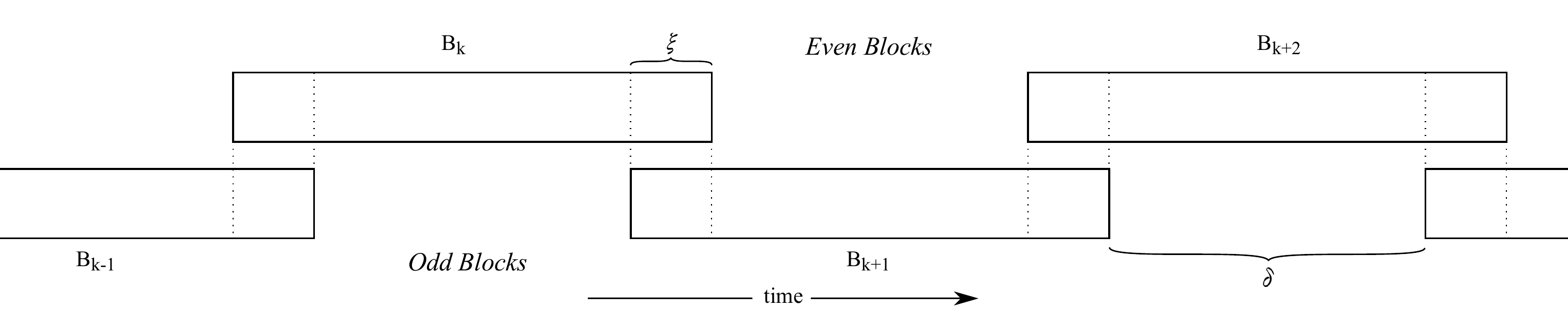}
	\end{center}
	\caption{A temporal blocking strategy with overlap $\xi$ and interior $\delta$ between blocks highlighted. The strategy will be efficient if the overlap $\xi$ is large enough to incorporate relevant information from neighbours.}
	\label{fig:temporal_strategy}
\end{figure}
{To improve mixing of blocked samplers in general it is often necessary to design a blocking strategy such that within-block correlation between variables is large while the correlation with out-of-block variables is small, see for example \cite{liu1994covariance} or \cite{turek2017automated}.} For state space models, this naturally implies blocking across time, and in Figure \ref{fig:temporal_strategy} a temporal strategy with overlap $\xi$ and interior $\delta$ is illustrated. We can in this case divide the blocks into even ($B_k$ of Figure \ref{fig:temporal_strategy} with even index $k$) and odd subsets such that each subset consists of non-overlapping blocks, see again the Figure \ref{fig:temporal_strategy}.  As analyzed in \cite{singh2017blocking} for blocked Gibbs samplers, temporal overlap leads to improved sharing of information across time and subsequent improved mixing. If the spatial dimension is very high, it can be necessary to block in the spatial domain as well; blocking strategies should in this case aim to exploit any spatial decorrelation if possible. 

%Recall that for matrices and spaces of matrices, the notation $M_{m,n}$ indicates the $m$'th row and $n$'th column, respectively. 
A few more remarks on notation: the restriction of $x \in M_{d, N}$ to a block $B = i{:}j \times l{:}m$ is the submatrix $x_B \in M_{|i{:}j|, |l{:}m|}$ corresponding to deleting all but the rows $i{:}j$ and columns $l{:}m$ of $x$. Similarly, the block restriction of $\nabla U$ is the map $\nabla_B U: M_{d,N} \rightarrow M_{|i{:}j|, |l{:}m|}$; the entries of the submatrix $\nabla_B U(x)$ are in correspondence with $\nabla U(x)$ via $\nabla_B U(x)_{a,b} = \nabla U(x)_{i+a-1, l+b-1}$. We extend this notation to also include the state and the velocity, with the submatrix under consideration being indicated by a subscript $B$. 

\section{Blocked bouncy particle sampler}
In this section we derive conditions under which the bouncy particle sampler of \cite{peters2012rejection, bouchard2018bouncy} can be run in blocked fashion; the resulting algorithm therefore applies to any target distribution $\pi$. If we assume that $\overline{B}$ consists of a single block of size $1{:}d \times 1{:}N$, the description below reduces to the standard bouncy particle sampler and it is therefore redundant to describe both.

The class of piecewise-deterministic Markov process we consider is a coupling of the solution $x(t)$ of the ordinary differential equation $dx(t)/dt = v(t)$, and a Markov jump process $v(t)$ where both transition operator $Q(v, dv)$ and rate process $\lambda(t)$ depends on $x(t)$ as well; $v(t)$ will henceforth be denoted the velocity. The joint process $(x(t), v(t))$ takes values in $M_{d,N} \times M_{d,N}$. Given an initialization $(x(0), v(0))$, the state flows as $(x(t), v(t)) = (x(0) + t \cdot v(0), v(0)),$ until an event $\tau$ is generated by an inhomogeneous Poisson process with rate $\lambda(t)$. At this point the velocity changes to $v(\tau) \sim Q(v(0), dv)$, and the process re-initializes at $(x(\tau), v(\tau))$. To use such a process for Markov chain Monte Carlo, the jump rate $\lambda(t)$ and transition kernel $Q$ of $v(t)$ are chosen such that the marginal stationary distribution of $(x(t))_{t \in [0, \infty)}$ is the target distribution of interest. {Exactly as in Metropolis-Hastings algorithms, we always want to move into regions of higher probability but desire to change direction, by a new choice of velocity vector, as we enter regions of declining probability.} This in turn implies that the rate is determined by the directional derivative of the energy $U$ in the direction of $v$, while the transition kernel $Q$ is a deterministic involution or random velocity change, for general details see \cite{vanetti2017piecewise}. 

\textit{Blocking of this process corresponds to a localization of the rate function and transition kernel such that each block is equipped with its own random clock and corresponding local velocity updating mechanism.} Subsequently, only velocities within a single block are changed at an event, while preserving the overall invariant distribution. In comparison with discrete time blocking that updates the variables one block at a time while keeping every other variable else fixed, in continuous time the block direction is changed while keeping every other direction fixed. For dimensions that are in multiple blocks, the additional clocks implies an excess amount of events compared to the base case of no overlap; the $\phi$ variable introduced below adjusts for this discrepancy by speeding up the velocity of the shared dimensions. Intuitively, as a dimension shared by $k$ blocks will have events $k$ times as often, it should move at $k$ times the speed to compensate. This also aligns exactly with discrete-time blocked sampling, where dimensions shared between blocks are updated twice as often. 

We now present the blocked bouncy particle sampler in detail. We assume that the velocity is distributed such that each $v_n^k \sim N(0,1)$ in stationarity{, and the stationary joint distribution of all velocities has density $p_v(v)$}. For a blocking strategy $\overline{B}$, we introduce the auxiliary variable $\phi \in M_{d,N}$ with entries
\[
\phi_n^k= \# \{ B \in \overline{B} \mid (k,n) \in B \},
\]
$\phi_n^k$ counts the number of blocks that include the $k$'th spatial dimension and $n$'th temporal dimension. Given $\phi$, the resulting block-augmented flow of the ordinary differential equation driving $x(t)$ is $t \mapsto x + t \cdot (\phi \star v)$; as mentioned, individual dimensions of $x$ are sped up in proportion to how many blocks includes them.  
With $x \mapsto \{x\}^+ \equiv \max \{0, x\}$, the rate function for the Poisson process $\Pi^B$ associated with block $B$ is
\[
\lambda_B(x,v) = \left \{ \langle \nabla_B U(x), v_B \rangle_F \right \}^+;
\]
the associated superposition of all blocks is the Poisson process $\Pi^{\overline B} = \cup_{B \in \overline B} \Pi^{B}$. Events generated by $\Pi^{\overline{B}}$ are denoted $\tau_b$ with $b$ referring to a bounce. Note that the inner product corresponds to the directional derivative $\partial U(x+t\cdot v)/\partial t$ restricted to $B$.  For the transition kernel, we define $\textsc{reflect}^B_x(v)$ as the (deterministic) reflection of the velocity $v_B$ in the hyperplane tangent to the block gradient at $x$:
\[
%\textsc{reflect}_x^B(v_B) 
v_B \leftarrow v_B - 2\frac{\langle \nabla_B U(x), v_B \rangle_F}{\Vert \nabla_B U(x) \Vert^2_F} \nabla_B U(x),
\]
while the remaining components of $v$ are {\it unchanged} by $\textsc{reflect}^B_x(v)$. (Note only the velocities that correspond to the block $B$ are updated.) Variable $v$ will also be updated by full velocity resampling via an independent homogeneous Poisson process with rate $\gamma$ to alleviate issues with reducibility, see \cite[Section 2.2]{bouchard2018bouncy}, and these event-times are denoted $\tau_r$ with $r$ referring to refreshment.
Without writing the refreshment operator, the infinitesimal generator of $(x(t), v(t))_{t \in [0,\infty)}$ is
\begin{equation}\label{generator}
\mathcal L^{\textsc{bBPS}} f(x,v) = \langle \nabla_x f(x, v), \phi \star v \rangle_F + \sum_{B \in \overline B} \lambda_B(x,v)\Big [f(x, \textsc{reflect}_x^B(v)) - f(x,v) \Big ],
\end{equation}
the sum of the block-augmented linear flow $\phi \star v$ driving $x(t)$ and the sum of Markov jump processes updating the block-restricted velocities $v_B$. 
\begin{proposition}\label{prop:bbps}
Consider a blocking strategy $\overline{B}$ and a target density $\pi(x) \propto \exp\{-U(x)\}$. With the generator defined in Equation (\ref{generator}), the blocked bouncy particle sampler has invariant distribution $\pi \otimes p_v$.
\end{proposition}
\begin{proof}
See section \ref{proof:bbps}.
\end{proof}
The most closely corresponding method to the blocked bouncy particle sampler is the factor algorithm presented in \cite[Section 3.1]{bouchard2018bouncy}. If the target distribution factorises over a finite set of individual factors $\overline{F}$ such that
\[
U(x) = \sum_{f \in \overline{F}} U_{f}(x_{f}),
\]
where $x_f$ corresponds to the restriction of the components in the factor, the local bouncy particle sampler of \cite{bouchard2018bouncy} can be applied. {Note that the derivation of the local bouncy particle sampler in \cite{bouchard2018bouncy}  is only considered under the above sum structure for the log density $U(x)$ and where each block is the complete set of variables $x_f$ for a factor. This contrasts with the blocked sampler, where blocks  are allowed to share variables arbitrarily and without the need for the energy to satisfy a sum structure.}
The blocked sampler algorithm in practice functions as a hybrid between the Zig-Zag sampler of \cite{bierkens2019zig} and the bouncy particle sampler: it incorporates the reflection operator when performing bounces, which allows for updating the velocity vector for multiple dimensions at event-times, but combines a more local rate structure akin to that of the Zig-Zag sampler. {In particular, if $|B| = 1$ for all $B \in \overline{B}$ and $\phi_n^k = 1$ for all $k, n \in [d] \times [N]$, the algorithm reduces to a process very close to the Zig-Zag sampler, with the velocity vector components ``flipping'' at their individual reflection event times (but an invariant normal distribution for the velocities compared to the binary uniform distribution of the standard Zig-Zag sampler.)} In this sense, the Zig-Zag sampler is naturally blocked, but does not allow for sharing of information across dimensions. In Algorithm \ref{bbps_algorithm} the blocked bouncy particle sampler is presented in an implementable form. 
\subsection{Simulation}\label{sec:event_generation}
Due to the simplicity of the flow the computational challenge of the algorithm is to generate correctly distributed event-times via Poisson thinning. The thinning procedures of \cite{lewis1979simulation} for simulating inhomogeneous Poisson processes is a two-step procedure that corresponds to finding a bounding process where direct simulation is available, and subsequently using rejection sampling to keep correctly distributed event-times. 

To employ thinning, local upper bounds $t \mapsto \bar{\lambda}_{B}(t)$ for each block needs to be estimated. For some fixed lookahead time $\theta > 0$ and current position $(x, v)$, local bounds satisfy
\begin{align}
\lambda_{B}(t) \leq \bar{\lambda}_{B}(t) \leq \ \max_{s \in [0, \theta)} \left \{\langle \nabla_B U(x + s (\phi \star v)), v_B \rangle_F \right \}^+,\quad \forall t \in [0, \theta)
\end{align}
and after $\theta$ time has passed, the bounds are recomputed at the new location $(x + \theta (\phi \star v),v)$, if no reflection or refreshment has occurred in the interrim. {The right-hand side is the worst-case bound and in all of our numerical examples we use this bound. In some particular cases, universal global bounds can be derived, but generally these bounds will have to be estimated by evaluating the rate function at some future time-point. If the blocks are individually log-concave densities, evaluating the rate at the lookahead time, $\lambda_B(\theta),$ gives a valid bound until an event occurs, or alternatively, one can apply the convex optimization procedure described in \cite[Section 2.3.1]{bouchard2018bouncy}.}
If blocks are overlapping, the local bounds of blocks in the neighbourhood $N(B)$ become invalid after a reflection and require updating. The generic process is given in Algorithm \ref{local_bounds}. Given the global bounding function
\begin{align}\label{global_bounding_function}
\bar \lambda_{\overline B}(t) = \sum_{B \in \overline B} \bar \lambda_B(t)
\end{align}
an event-time $\tau$ is simulated from $\Pi_{\bar \lambda_{\overline B}}$, a block $B$ is selected with probability proportional to its relative rate $\bar \lambda_B(\tau)/\bar \lambda_{\overline B}(\tau)$, and finally a reflection is carried out with probability corresponding to the true rate function relative to the local bound $\lambda_B(\tau)/\bar \lambda_B(\tau)$. Given the local rate functions, the dominant cost is the unsorted proportional sampling of a block, which is done in $O(|\overline B|)$. We propose to choose $\theta$ such that the expected number of events generated by the bounding process on the interval $[0, \theta]$ is equal to 1, as we can always decrease the computational cost of calculating bounds by changing $\theta$ to another value that brings the expectation closer to 1. In our numerical examples we have tuned $\theta$ to approximately satisfy this requirement.

\begin{algorithm}
	\DontPrintSemicolon
	\KwData{Initialize $(x^0, v^0)$ arbitrarily, set instantaneous runtime time $t = 0$, index $j = 0$, total execution time $T > 0$, lookahead time $\theta > 0$ and valid bound time $\Theta = \theta$.} \tcp{Variable $t$ denotes instantaneous runtime, $\theta$ is lookahead time for computing Poisson rate bounds and $\Theta > t$ are integer multiples of $\theta$. 
}
	
	$(\bar{\lambda}_{B})_{B\in  \overline B} \leftarrow \textbf{LocalBounds}(x^0, v^0, \theta,  \overline B)$ \tcp{Calculate initial bounds}
	\While{$t \leq T$}
	{
		$j \leftarrow j + 1$ \;
		\indent $\tau_b \sim \text{Exp}(\sum_{B}  \bar{\lambda}_{B} )$ \tcp{Reflection/Bounce time}
		\indent $\tau_r \sim \text{Exp}(\gamma)$ \tcp{Refreshment time} 
		$\tau^j \leftarrow \min \{ \tau_r, \tau_b\}$\;
		\If{$\tau^j + t > \Theta$}
		{
			\tcp{Runtime+event time exceeds valid time for bound, reinitalize at $\Theta$}
			$x^j \leftarrow x^{j-1} + (\Theta - t)\cdot  \phi \star v^{j-1} $ \;
			$v^j \leftarrow v^{j-1}$ \;
			$(\bar{\lambda}_{B})_{B \in  \overline B} \leftarrow \textbf{LocalBounds}(x^j, v^j, \theta,  \overline B)$ \;
			$t \leftarrow \Theta$ \tcp{Update runtime}
			$\Theta \leftarrow \Theta + \theta$ \tcp{New valid bound time}
		}
		\Else
		{
			$t \leftarrow t + \tau^j$ \tcp{Update runtime}
			$x^j \leftarrow  x^{j-1} + \tau^j \cdot \phi \star v^{j-1}  $\;
			\If{$\tau^j < \tau_r$}{

				\tcp{Select block for reflection}
				Draw $B \in \overline{B}$ with $\mathbb{P}(B = B_i) = \bar{\lambda}_{B_i}(\tau^j) / \bar{\lambda}_{\overline{B}}(\tau^j)$ \;
				$u \sim \text{U}[0,1]$ \;
				\If{$u < \lambda_{B}(x^j, v^{j-1}) / \bar{\lambda}_{B}(\tau^j)$}
				{
					$v_{B}^j \leftarrow \textsc{reflect}_x^B v^{j-1}_{B}$ \;
					\tcp{Update bounds for blocks affected by reflection}
					$(\bar{\lambda}_{B'})_{B'\in N(B)} \leftarrow \textbf{LocalBounds}(x^j, v^j, \Theta - t, N(B))$ \;
				}
				 \Else 
				{
				 	$v^j \leftarrow v^{j-1}$
			 	}
			} \Else 
			{
				\tcp{Refresh all velocities}
				$v^j \sim \mathcal{N}(0, I_{(d \times T) \times (d \times T)})$\;
				$(\bar{\lambda}_{B})_{B \in  \overline B} \leftarrow \textbf{LocalBounds}(x^j, v^j, \theta,  \overline{B})$ \;
			}
		}
		\Return{$(x^j, v^j, \tau^j)$}
	}
\caption{Blocked Bouncy Particle Sampler}	\label{bbps_algorithm}
\end{algorithm}

\begin{algorithm}
	\DontPrintSemicolon
	\caption{\bf{LocalBounds}($x, v, \theta,  \overline B$)}\label{local_bounds}
	\KwData{$(x, v)$, $\theta > 0$ and set of blocks $ \overline B$.}
	\For{$B \in \overline B$}
	{
		Find function $t \mapsto \bar{\lambda}_{B}(t)$ that on $[0, \theta)$  satisfies \;
		$\bar{\lambda}_{B}(t) \leq \max_{s \in [0, \theta)} \lambda_{B}(x + s  (\phi \star v), v_{B})$.
	}
	\Return{$(\bar{\lambda}_{B})_{B \in  \overline B}$}
\end{algorithm}

\section{Parallel velocity updates via partitioned blocking strategies}
As mentioned in the introduction, \cite{singh2017blocking} shows that the even-odd blocking strategy with overlaps is known to improve mixing, and furthermore allows for parallelization of updates in the case of Kalman smoothers or particle filter-based smoothing algorithms. Conversely, the current crop of piecewise-deterministic Markov process-based samplers are all purely sequential, in the sense that at each event-time only the velocity of a single factor or dimension is updated, and these samplers therefore fail to exploit any conditional independence structure available. We will in this section provide an alternative implementation {(see Algorithm \ref{conditional_algorithm})} of the blocked bouncy particle sampler that mimics the even-odd strategy of discrete-time blocked samplers, extends to the fully spatially blocked setting, and allows for parallel implementation of updates at event-times. 
To utilise this method, we need a partition of the blocking strategy into sub-blocking strategies such that no two blocks in any sub-blocking strategy share any variables. To this end, we capture the no-overlap condition precisely in the following assumption: 
\begin{assumption}
\label{def:propdisjoint}
Consider a blocking strategy $\overline B$. We will assume given a partition $\cup_{k=1}^K \overline B_k = \overline B$ of the blocking strategy that satisfies, for each sub-blocking strategy $\overline B_k, k = 1,2,\ldots K$ and for all blocks $B, B' \in \overline B_k$ such that $B \neq B'$, that
\[
B \cap B' = \emptyset.
\]
\end{assumption}
This assumption also applies to fully spatiotemporal blocking schemes and not just temporal strategies. We will for illustrative purposes only describe in detail the simplest even-odd scheme of temporal blocking, which corresponds to $K=2$ sub-blocking strategies such that no blocks that are temporally adjacent are in the same sub-blocking strategy. {As shown in Figure \ref{fig:temporal_strategy}, each block is assigned a unique integer number $k$. We then partition the strategy into two sets of blocks based on whether $k$ is an even or odd integer, and denote the sub-blocking strategies  $\{ \overline B_{odd}, \overline B_{even} \}$. In Figure \ref{fig:temporal_strategy} we illustrate such a strategy, where the top row shows the even blocks, and the lower row the odd blocks. Note that individual even blocks have no state variables in common (similarly for individual odd blocks). Furthermore, for a Markovian state space model, each block is chosen to be a consecutive time sequence of states.}

%one can observe that the largest valid upper bound amongst all blocks  must be a valid bound for each block in the sub-blocking strategy. This implies that the algorithm subsequently will have two exponential clocks with rates

{For such a sub-blocking strategy, we then find the maximum rate among all blocks inside a sub-blocking strategy
\begin{align}\label{lambda_bar}    \widehat{\Lambda}_{odd}(x,v) = \max_{B \in \overline B_{odd}} \lambda_B(x,v), \qquad \widehat{\Lambda}_{even}(x,v) = \max_{B \in \overline B_{even}} \lambda_B(x,v)
\end{align}
and denote their associated Poisson processes $\Pi_{odd}^{\overline B}$ and $\Pi_{even}^{\overline B}$. By construction, we will have two exponential clocks, one for  the set of blocks $\overline B_{odd}$ and one for $\overline B_{even}$.
To detail what happens at an event time, consider an event generated by the superposition of $\Pi_{odd}^{\overline B}$ and $\Pi_{even}^{\overline B}$ and say $\Pi_{odd}^{\overline B}$ generated the event. Then for each block $B \in \overline B_{odd}$, the following kernel $Q_x^B(v, dv)$ is used to update the velocity of that block 
\[
Q_x^B(v, dv) = \delta_{\textsc{reflect}_x^B (v)}(dv) \frac{\lambda_B(x,v)}{\widehat{\Lambda}_{odd}(x,v)} + \delta_v(dv) \left (1-\frac{\lambda_B(x,v)}{\widehat{\Lambda}_{odd}(x,v)} \right).
\]
 This simultaneous velocity update of all the blocks in the particular set of blocks is permissible since the blocks of each set have no states in common, i.e. do not overlap. In Algorithm \ref{conditional_algorithm} we provide pseudo-code for a fully implementable version of the blocked Bouncy Particle Sampler under an even-odd partition of the blocking strategy.}

We will show invariance for the  particular case considered above;  the result holds in general for any partition satisfying Assumption \ref{def:propdisjoint}.

\begin{proposition}\label{prop:pibps_invariance}
Let $\{ \overline B_{odd}, \overline B_{even} \}$ be a temporal strategy for $\pi$ and $\overline B$ satisfying Assumption \ref{def:propdisjoint}. Then the Markov process with associated generator
\[
\mathcal L^{\textsc{eoBPS}}f(x,v) =  \langle \nabla_x f(x,v), \phi \star v \rangle_F +  \sum_{\kappa \in \{odd, even \}}  \widehat \Lambda_\kappa(x,v) 
    \left [ 
    \sum_{B \in \overline B_\kappa} 
        \int \left [
         f(x, v') - f(x,v) 
        \right  ]Q_x^B(v, dv')
    \right ]
\]
has invariant distribution $\pi \otimes p_v$, where $\widehat\Lambda_\kappa(x,v)$ is defined in Equation \ref{lambda_bar}.
\end{proposition}
\begin{proof}
See section \ref{proof:pibps_invariance}.
\end{proof}
In contrast to the basic blocked BPS, the generator of Proposition \ref{prop:pibps_invariance} has a single overall event-time generated from sum of odd and even strategies, but multiple overlapping event-times for the blocks contained in the sub-blocking strategy that generated the event. The even-odd algorithm therefore corresponds to an implementation that ``lines up'' the event times in such a way that is beneficial for a parallel implementation. Relative to the blocked bouncy particle sampler, the even-odd implementation iterates over every block in the sub-blocking strategy that generated the event, updating velocities of the blocks with probability proportional to the ratio of the block's rate $\lambda_B$ evaluated at the current state $(x,v)$ to the rate of the sub-blocking strategy given by the max-bound. It therefore becomes possible to parallelize the updating step, for example with multiple processors allocated to each sub-blocking strategy, say one processor per block of the sub-blocking strategy. In contrast to the generator in Proposition \ref{prop:bbps}, the event rate of the sampler in Proposition \ref{prop:pibps_invariance} is now the maximum over the rates in a sub-blocking strategy which should grow slower than the sum rate in Proposition \ref{prop:bbps} as the global dimension $(d)$ and thus number of blocks grow.

 If the spatial dimension is significant, it will be necessary to also carry out spatial blocking. Under a full spatiotemporal strategy, the above implementation naturally extends to a four clock system, consisting of alternating even-odd temporal strategies over each `row' of spatial blocks,  such that that no blocks from the same sub-strategy overlap; this in turn guarantees that Assumption \ref{def:propdisjoint} is satisfied. 

In practice, $\widehat{\Lambda}_\kappa$ is not available, as we can not evaluate the gradient in continuous time. Similarly to Algorithm \ref{bbps_algorithm}, we employ a lookahead time $\theta$ and a trivial global bound for the Poisson rates that is valid for the interval $(t,t+\theta]$ where as before $t$ is the instantaneous runtime. 
For any fixed $\theta > 0$, assuming $(x(t),v(t))=(x,v)$, let the globally valid bound $\bar \Lambda_\kappa$, with $\kappa \in \{odd, even\}$, be given as
\begin{align*}
\bar \Lambda_\kappa   \equiv \max_{B \in \overline B_{\kappa}} \sup_{s \in [0, \theta]} \lambda_B \left (x + s \cdot (\phi \star v), v\right).
\end{align*} 
As in Algorithm \ref{bbps_algorithm}, we use this rate to define a bounding Poisson processes and apply thinning to find the appropriate events, see line \ref{algo_reflection_time} in the algorithm.
{In practical implementations of piecewise-deterministic algorithms, tighter bounds for the event-times are in general necessary to avoid wasteful computation from false events. Our max-type bound is tighter than the sum-type bound, and we can therefore have a larger lookahead time $\theta$. (Again, $\theta$ should be chosen such that the expected number of events generated by the bound in an interval of size $\theta$ is 1.) With the max-type bound, the even-odd implementation will have larger event times compared to the blocked BPS.}

%In the extreme case when the individual max bound for each block is comparable, the even-odd implementation will have event times of the order $O(|\overline{B}|)$ larger.} Overall, this leads to a constant order cost reduction for the even-odd implementation via fewer calls to \textsc{LocalBounds.

{The growth of the rate of the max-type bound, as a function of the number of blocks, is studied in the following result.} In particular, under plausible assumptions on the tail-decay of the target distribution we can bound the expected rate.

\begin{lemma}\label{lemma:bound_rate}
Assume that for all $B \in \overline{B}$ 
\[
\mathbb P(\lambda_B(x,v) > s) \leq 2e^{-2\alpha s}
\]
for some $\alpha > 0$. Then both the odd and even sub-blocking strategies, indicated by subscript $\kappa$, satisfies
\[
\mathbb E_\pi \widehat \Lambda_{\kappa}(x,v) \leq \frac{2e}{\alpha} \log |\overline{B}_{\kappa}|
\]
\end{lemma}
\begin{proof}
See section \ref{proof:event_rates}.
\end{proof}
{In the Gaussian case, the rate function is a mixture of a point-mass at zero and a density proportional to the modified Bessel function of the second kind with order depending on the dimension, and this function is known to have sub-exponential decay for any $d$, see for example \cite{yang2017approximating}. We note that the key point of Lemma \ref{lemma:bound_rate} is to verify that utilizing the maximum over blocks does not lead to pathological behaviour.}

To elaborate on the computational costs of the samplers, we compare the cost to run the samplers for one second. The exponential event-times of Poisson processes indicates we can expect $O(\log|\overline{B}_{\kappa}|)$ events per time unit (line \ref{conditional_algorithm}.\ref{algo_reflection_time}) via $\widehat{\Lambda}_\kappa$, each costing $O(|\overline{B}_{\kappa}|)$ evaluations of blocks (Line \ref{conditional_algorithm}.\ref{algo_partition_calculation}) per kernel $Q_x^B$. Thus the total cost of the even-odd sampler per second is $O(|\overline{B}_{\kappa}|\log|\overline{B}_{\kappa}|)$. {In comparison, the local bouncy particle sampler has a rate function defined as $\Lambda_{\overline{F}} = \sum_{F \in \overline{F}} \lambda_F(x,v) = \sum_{F \in \overline{F}} \max \{ 0, \langle \nabla U_F(x), v \rangle \}$, with $\overline{F}$ is the set of factors of $U$, $\nabla U_F(x)$ is the gradient of the factor $U_F(x)$. In this case,} the event rate growth is of the order $O(|\overline{F}|)$ by the inequality
\[
\mathbb E_\pi \Lambda_{\overline{F}}(x,v) = \mathbb E_\pi \sum_{F \in \overline{F}} \lambda_{F}(x,v) \geq |\overline{F}| \min_{F \in \overline{F}} \mathbb E_\pi \lambda_{\overline{F}}(x,v),
\]
combined with $O(1)$ costs per event-time, for a total cost of $O(|\overline{F}|)$ per sampler second. However, we note again that each of the $O(|\overline{B}_{\kappa}|)$ evaluations of the blocks can be carried out fully in parallel as no velocities are shared across ringing blocks. {Furthermore, in the continuous-time Markov Chain Monte Carlo literature the metric of effective sampler size per unit of trajectory length has been considered, and it is at this stage not known theoretically how the blocked BPS and the local BPS differ under this alternative measure of efficiency.}

%\input{even_odd_algo}
%%%%%%
%%%%%%
\begin{algorithm}
	\DontPrintSemicolon
	\KwData{Initialize $(x^0, v^0)$ arbitrarily, set instantaneous runtime $t = 0$, index $j = 0$, total execution time $T > 0$, lookahead time $\theta > 0$ and  bound time $\Theta = \theta$.} 
	$(\bar{\lambda}_{B})_{B\in \overline B} \leftarrow \textbf{LocalBounds}(x^0, v^0, \theta, \overline B)$\tcp{Calculate initial bounds}

	\While{$t \leq T$}
	{
		$j \leftarrow j + 1$ \;
		\indent $\tau_b \sim \text{Exp}(\sum_{i \in \{even, odd\}} \bar \Lambda_i)$ \tcp{Reflection/bounce time} \label{algo_reflection_time}
		\indent $\tau_r \sim \text{Exp}(\gamma)$ \tcp{Refreshment time} 
		$\tau^j \leftarrow \min \{ \tau_r, \tau_b\}$\;
		\If{$\tau^j + t > \Theta$}
		{
			\tcp{Runtime+event time exceeds valid time for bound, reinitalize at $\Theta$}
			$x^j \leftarrow x^{j-1} + (\Theta - t)\cdot  \phi \star v^{j-1} $ \;
			$v^j \leftarrow v^{j-1}$ \;
			$(\bar{\lambda}_{B})_{B \in \overline B} \leftarrow \textbf{LocalBounds}(x^j, v^j, \theta, \overline B)$ \;
			$t \leftarrow \Theta$ \tcp{Update runtime}
			$\Theta \leftarrow \Theta + \theta$ \tcp{New valid bound time}
		}
		\Else
		{
			$t \leftarrow t + \tau^j$ \tcp{Update runtime}
			$x^j \leftarrow  x^{j-1} + \tau^j \cdot \phi \star v^{j-1}  $\;
			\If{$\tau^j < \tau_r$}{

				\tcp{Select blocking strategy subset}
				Draw $\kappa \in \{even, odd \}$ with $\mathbb{P}(\kappa = i) \propto \bar \Lambda_i$ \;
				\For{$B \in \overline B_\kappa$ \label{algo_partition_calculation}} 
				{
				    $u  \sim \text{U}[0,1]$ \;
				    \If{$u < \lambda_{B}(x^j,v^{j-1} ) / \bar \Lambda_\kappa$}
				{
					$v_{B}^j \leftarrow \textsc{reflect}^{B}_x v^{j-1}_{B}$ \;
					\For{$B \in \{B_{\max\{j-1, 1\}}, B_j, B_{\min \{j+1, |\overline B|\}} \}$} 
					{
						$\bar{\lambda}_{B} \leftarrow \textbf{LocalBounds}(x^j, v^j, \Theta- t, B)$ \;
					}

				}
				 \Else 
				{
				 	$v^j_B \leftarrow v^{j-1}_B$
			 	}
				}

			} \Else 
			{
				\tcp{Refresh all velocities}
				$v^j \sim \mathcal{N}(0, I_{(d \times T) \times (d \times T) })$\;
				$(\bar{\lambda}_{B_i})_{i\in \overline B} \leftarrow \textbf{LocalBounds}(x^j, v^j, \Theta-t, \overline{B})$ \;
			}
		}
	}
\caption{Even-odd implementation of blocked Bouncy Particle Sampler}
\label{conditional_algorithm}
\end{algorithm}

%%%%%
%%%%%

\section{Numerical Experiments}
We will in the following two sections provide two numerical experiments comparing the blocked BPS, the local BPS, and particle Gibbs. As we are primarily interested in latent state estimation, we have not considered parameter inference in the examples below. A natural approach here would be to run a Metropolis-within-Gibbs sampler that proposes an update to the parameter vector after, for example, running the continuous-time sampler for a second. The proposal of the parameter vector could be done in discrete time, or alternatively using the BPS for parameter vector. This latter strategy was proposed for continuous-time Markov chains in \cite{zhao2019analysis}. Alternatively, the parameters could be inferred jointly in continuous-time together with the latent states; the parameter vector could be included in the blocking strategy, in particular if the parameter vector is also dynamic across time. 

\subsection{Linear Gaussian toy model}
We consider an autoregressive model of order 1 given by
\begin{align*}
    x_n &= Ax_{n-1} + \eta_n, \quad \eta_n \sim \mathcal N(0,I_d)  \\
    y_n &= x_n + \epsilon_n, \quad \epsilon_n \sim \mathcal N(0,I_d)
\end{align*}
with $A$ an autoregressive matrix with entries $A_{ij} = \text{kern}(i,j)/\left (\psi + \sum_{l=1}^d \text{kern}(i,l)\right )$ with $\text{kern}(i,j) = \exp \left \{-\frac{1}{2\sigma^2} |i-j|^2 \right \}$ and $\psi > 0$ a constant, and finally, $x_0 \sim \mathcal N(0, I_d)$.
First, we want to compare the empirical mixing speed of the blocked and factor bouncy particle samplers. We consider a simulated model of $d = 3$ and $N=1000$, $\sigma^2 = 5$, and $\psi = 0.1$. We initialise each sampler at the zero vector, run until $T = 1000$, and thin at every 0.1 sampler second. In Table \ref{example_one_table} we provide detailed specifications of the setups for the various algorithms and results from a representative run of the algorithms.

\begin{table}[ht]
\begin{centering}
\begin{tabular}{lrrr}
\toprule
\textit{Algorithm}                & Local BPS & Blocked BPS & Even-odd \\ \midrule
Dimensions per factor/block       & 60        & 60          & 60            \\
Number of factors/blocks          & 50        & 101         & 101           \\
Number of sub-blocking strategies & -         & -           & 2             \\ \midrule
Temporal width                    & 20        & 20          & 20            \\
Spatial width                     & 3         & 3           & 3             \\
Temporal overlap                  & -         & 10          & 10            \\
Spatial overlap                   & -         & 0           & 0             \\ \midrule
Relative performance              & 0.48      & 0.67        & 1.00          \\ \bottomrule
\end{tabular}
\caption{Specification of implementations and results for the autoregressive Gaussian model with $T = 1000$ and $d = 3$. Performance is measured in terms of ESS/s relative to the even-odd bBPS.}
\label{example_one_table}
\end{centering}

\end{table}

In Figure \ref{fig:mse_v_time}, we plot the log of the mean square error as a function of time for increasing block overlap; empirically the blocked sampler with block width 20 and overlap 10 reaches stationarity around 3 times faster than the factor version. In Figure \ref{fig:msjd_v_dim}, we compare the mean squared jumping distance of the first spatial dimension after discarding the first 25\% of samples. For the overlapping sections, the exploration is, due to the shared overlap and $\phi$, happening at twice the speed, and, accordingly, four times the mean-square jumping distance compared to the factor algorithm. In terms of effective sample size per second, the blocked and even-odd samplers are about 30-40\% and 100\% more efficient respectively than the factor sampler, without using any parallel implementation. It is observed in general for any choice of $d$ and $T$ that the benefits of speeding up the dimensions compensate for the increased computational cost due to the overlaps. We also note that for models like this where the spatial dimension is low, there is not a strong argument to use PDMP-based methods as particle Gibbs with a basic particle filter will be more than adequate.

\begin{figure}
\begin{subfigure}{.45\textwidth}
	\begin{center}
	\includegraphics[width=\textwidth]{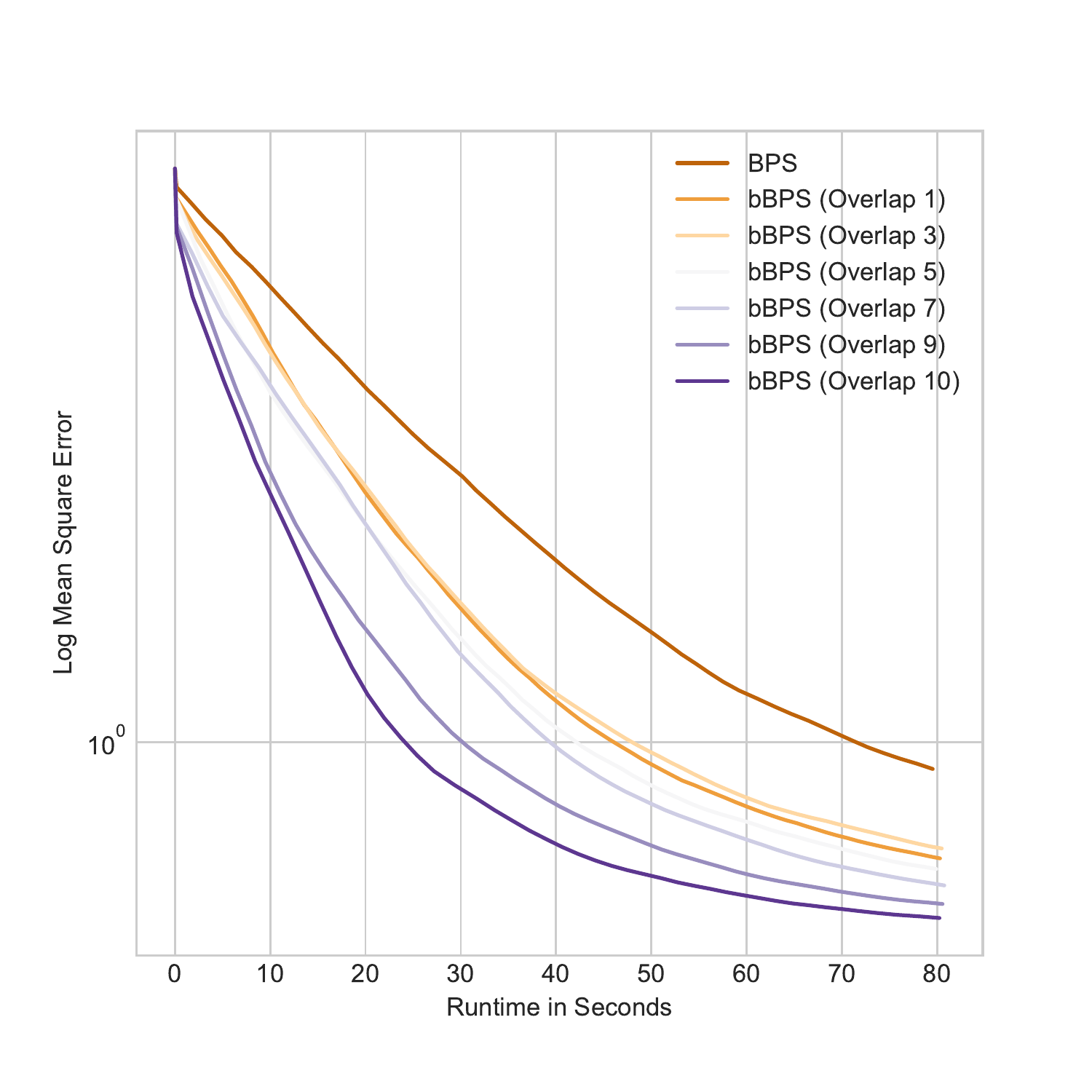}
	\end{center}
	\caption{}
	\label{fig:mse_v_time}
\end{subfigure}
\begin{subfigure}{.45\textwidth}
	\begin{center}
	\includegraphics[width=\textwidth]{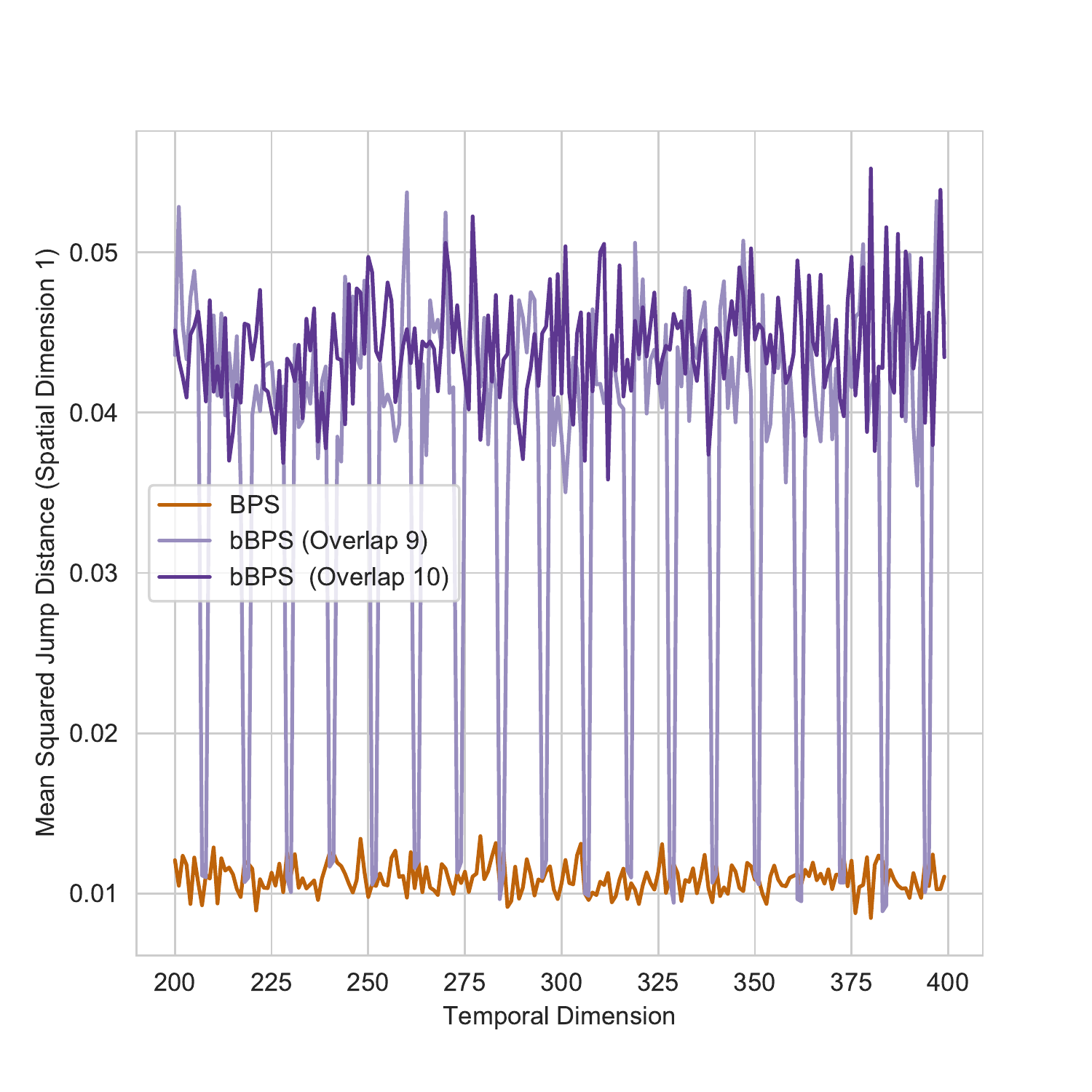}
	\end{center}
	\caption{}
	\label{fig:msjd_v_dim}
\end{subfigure}\vfill
\vspace{-0.5\baselineskip}
\iffalse
\begin{subfigure}{.45\textwidth}
	\begin{center}
	\includegraphics[width=\textwidth]{gaussian_traceplot.pdf}
	\end{center}
	\caption{}
	\label{fig:gaussian_trace}
\end{subfigure}
\begin{subfigure}{.45\textwidth}
	\begin{center}
	\includegraphics[width=\textwidth]{gaussian_acfplot.pdf}
	\end{center}
	\caption{}
	\label{fig:gaussian_acf}
\end{subfigure}
\fi
\caption{(a) Mean square error estimate per unit of CPU time of the autoregressive Gaussian model as the overlap varies. (b) Mean square jump distance for the standard bouncy particle sampler and blocked counter-part with overlaps 9 and 10, showcasing the impact $\phi$ has on exploration. In particular, the dips for the overlap 9 case corresponds to the variables that are part of a single block only, and subsequently are not sped up. We show a subset of 200 time points to enhance detail.}
%(c) Traceplot of $U(x_{1:N})$ for the autoregressive Gaussian model, $d = 50$  and $N = 60$ from 50,000 samples. (d) Autocorrelation plot of $\pi(x_{1:N})$ for the first 200 lags. }
\end{figure}

Second, we consider the case where $d=200$ and $T=100$ to illustrate the benefits of spatial blocking in high dimensional scenarios. In this case we also include a spatiotemporal blocking strategy, and the details of the example and a representative simulation are provided in Table \ref{example_two_table}. The model and example parameters are otherwise as described above.

% \usepackage{booktabs}
\begin{table}[ht]
\centering
\begin{tabular}{@{}lrrrr@{}}
\toprule
\textit{Algorithm}                & Local BPS & Blocked BPS & Even-odd & Spatiotemporal \\ \midrule
Dimensions per factor/block       & 400       & 400         & 400           & 54                  \\
Number of factors/blocks          & 50        & 99          & 99            & 957                 \\
Number of sub-blocking strategies & -         & -           & 2             & 4                   \\ \midrule
Temporal width                    & 2         & 2           & 2             & 9                   \\
Spatial width                     & 200       & 200         & 200           & 6                   \\
Temporal overlap                  & -         & 1           & 1             & 3                   \\
Spatial overlap                   & -         & 0           & 0             & 2                   \\ \midrule
Relative performance              & 0.36      & 0.34        & 0.56              & 1.00                   \\ \bottomrule
\end{tabular}
\caption{Specification of implementations and results for the autoregressive Gaussian model with $T = 100$ and $d = 200$. Performance is measured relative to ESS/s for the spatiotemporal bBPS.}
\label{example_two_table}
\end{table}

The spatiotemporally blocked sampler significantly outperforms the other implementations, with effective sample size per second typically 2-4 times larger, evidenced over multiple runs with random trajectories generated under the model. The even-odd temporal implementation blocked strategy is often still efficient even when the number of dimensions per block is up to 400, but the relative ESS/s is on aggregate lower than the spatiotemporally blocked version. Furthermore, this discrepancy will only increase under models with even higher spatial dimension. As before, no concurrent implementation was used, indicating that additional improvements in performance are possible for the partitioned blocking schemes when parallelized over multiple processors.

\subsection{Heavy-tailed stochastic volatility with leverage effects}
We will in this section consider an example based a stochastic volatility model of the Dow Jones Industrial Average (DJIA) equity index to explore the efficiency of the even-odd implementation of the BPS in comparison with two benchmark implementations of particle Gibbs when the spatial dimension is moderate and the length of the time-series is long. 
Stochastic volatility models are widely used in finance and econometrics. They model the volatility of financial assets as a dynamic latent process to capture the time-varying and persistent nature of changes in asset returns. We will analyse a general model proposed by \cite{ishihara2012efficient} that incorporates heavy-tailed observations and leverage effects, see \cite{cont2001empirical} for empirical discussion of these effects. To test the blocked algorithms on a reasonably challenging dataset, we attempt to estimate the latent volatility of the 29 continuously available constituents of the DJIA between April 1st 2017 and April 6th 2020, for a total of $29 \times 757 = 21953$ latent states. This period is characterized both by relatively low volatility and historical high levels uncertainty due to the COVID-19 pandemic, see \cite{world2020coronavirus} for example. 

\begin{figure}
\begin{center}
\includegraphics[width=\textwidth]{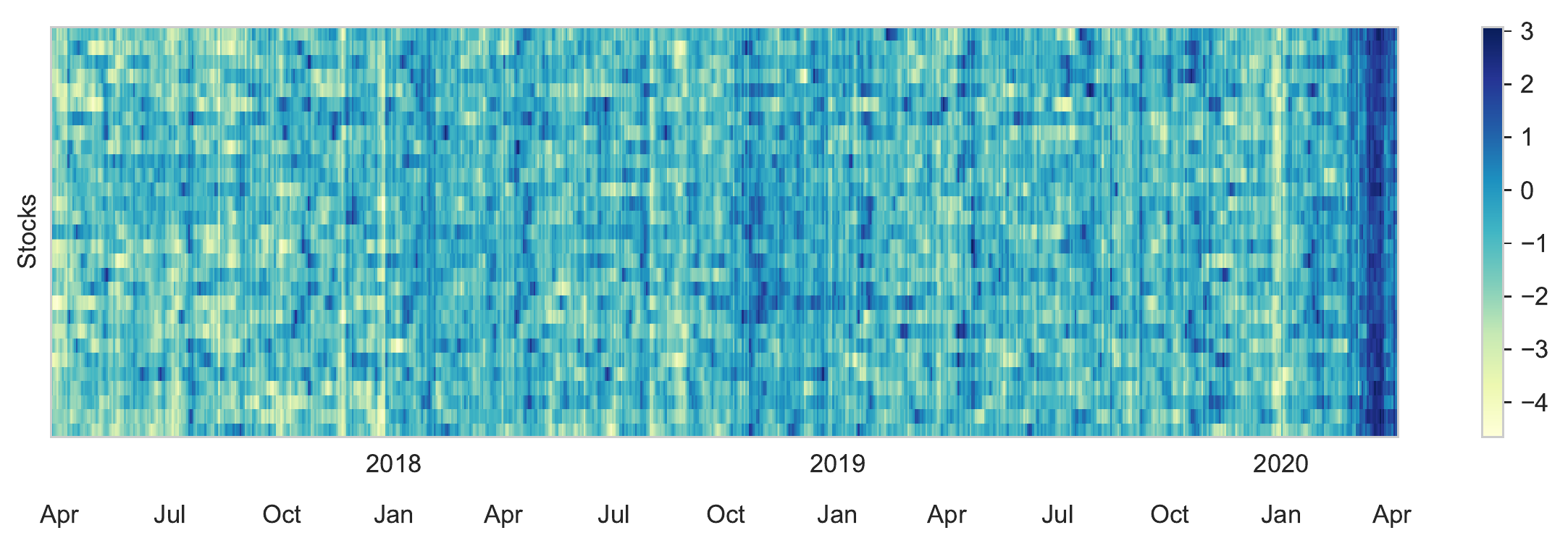}
\end{center}
\label{fig:estimated_vol}
\caption{Estimated latent volatility (the posterior mean after discarding the first 250 samples) via the blocked bouncy particle sampler for the 29 continuously available constituents of the Dow Jones Industrial Average 30 index between April 1st 2017 and April 6th 2020. }
\end{figure}

Let $x_n \in \mathbb{R}^d$ be an unobserved vector of volatilities, and $y_n \in \mathbb{R}^d$ be observed asset log returns. The dynamics over a fixed time horizon $n = 1,2,\ldots, N$ are
\begin{eqnarray*}
x_{n+1} &=& A x_{n} + \eta_n \\
y_n &=& \gamma_n ^{-\frac{1}{2}}\Lambda_n \epsilon_n, \qquad \Lambda_n = \text{diag}(\exp \Big \{ \frac{x_n}{2} \Big \})
\end{eqnarray*}
with $A =\text{diag}(\alpha_1, \alpha_2, \ldots, \alpha_d)$, where $\alpha_i \in [0, 1), \forall i \in \{1, 2, \ldots, d\}$. The noise is jointly modelled as 
\[
\begin{pmatrix}
\eta_n \\
\epsilon_n
\end{pmatrix} \sim \mathcal{N}(0, \widehat{\Sigma}), \text{ with } \widehat{\Sigma} = 
\begin{pmatrix}
\Sigma_\eta & \Sigma_{\rho} \\
\Sigma_{\rho} & \Sigma_\epsilon
\end{pmatrix} 
\]
and $\widehat{\Sigma}$ a $2d \times 2d$ matrix. The off-diagonal block matrices introduce leverage effects in the model if they are negative definite. Finally, for some $\nu \in \mathbb{N},$ $\gamma_n \sim \Gamma(\frac{\nu}{2}, \frac{\nu}{2})$ is a memory-less stochastic process independent of $(\eta_n, \epsilon_n)$. The resulting observation noise is multivariate t-distributed with $\nu$ degrees of freedom, details are in \cite{ishihara2012efficient}. For the initialization we assume that $x_1 \sim 
\mathcal N(0, (I_d - AA)^{-1} \Sigma_{\eta})$. Define $y_n^\gamma = \sqrt{\gamma_n} y_n$ as the observations whenever $\gamma_n$ is known; it follows that $y_n^\gamma = \Lambda_n \epsilon_n$ and inference can be carried out with this observation sequence instead. Conditional on $\gamma_{1:N}$ and using basic properties of multivariate Gaussians, the transition distributions can be written as
\begin{eqnarray*}
p(x_n|x_{n-1}, y_{n-1}^\gamma) &=& \mathcal N(Ax_{n-1} + \Sigma_\rho \Sigma_\epsilon^{-1} \Lambda_{n-1}^{-1} y_{n-1}^\gamma, \Sigma_\eta - \Sigma_\rho\Sigma_\epsilon^{-1}\Sigma_\rho) \\
p(y^\gamma_n|x_n) &=& \mathcal N(0, \Lambda_n \Sigma_\epsilon \Lambda_n),
\end{eqnarray*}
implying that the distribution has a more complicated dependence structure, as the past observation feeds into the next realized state. Furthermore, the state transition is non-linear in the previous state variable due to the leverage effect.

\begin{figure}
\begin{subfigure}{.45\textwidth}
	\begin{center}
	\includegraphics[width=\textwidth]{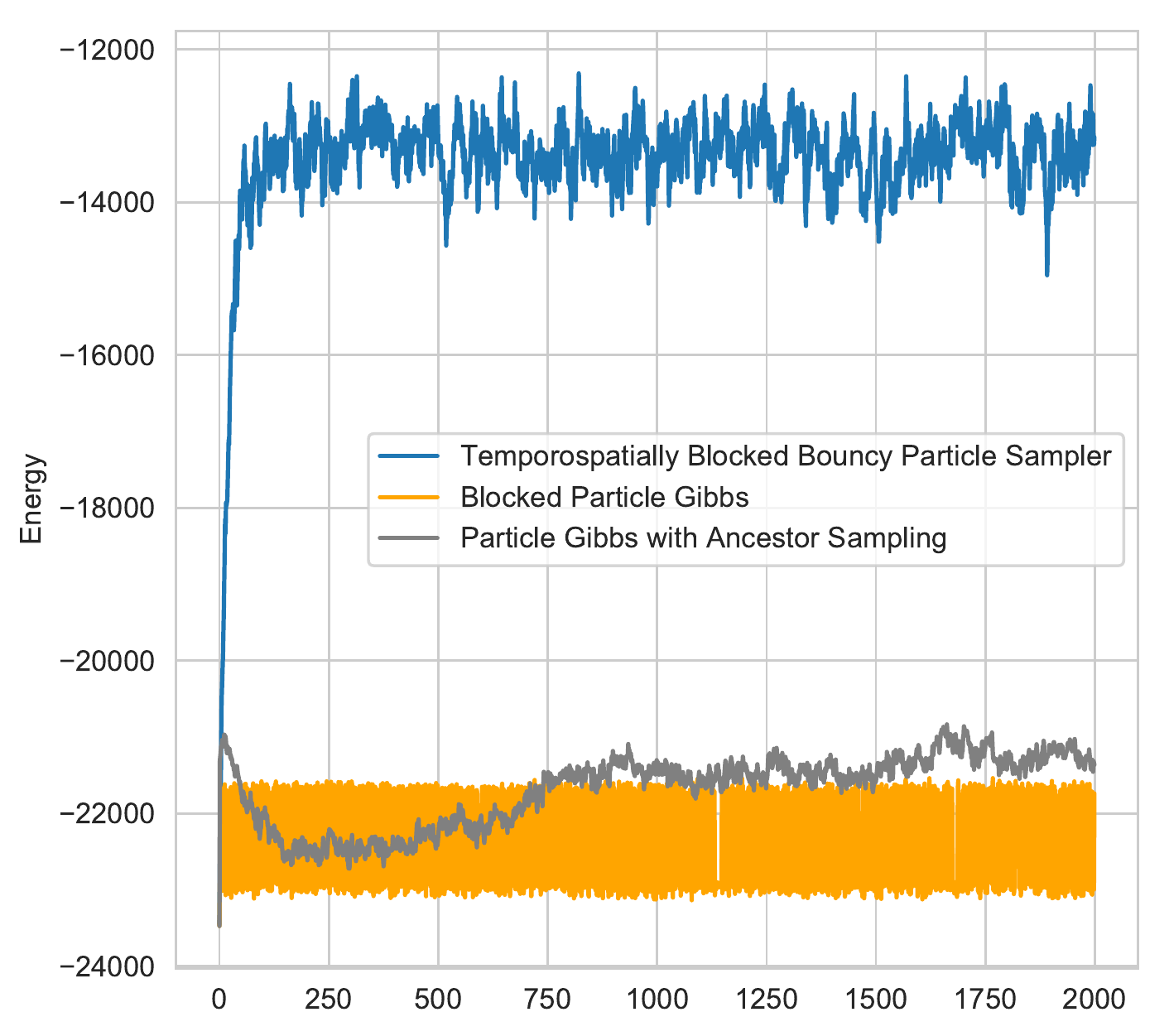}
	\end{center}
	\caption{}
	\label{fig:stochvol_traceplot}
\end{subfigure}
\begin{subfigure}{.45\textwidth}
	\begin{center}
	\includegraphics[width=\textwidth]{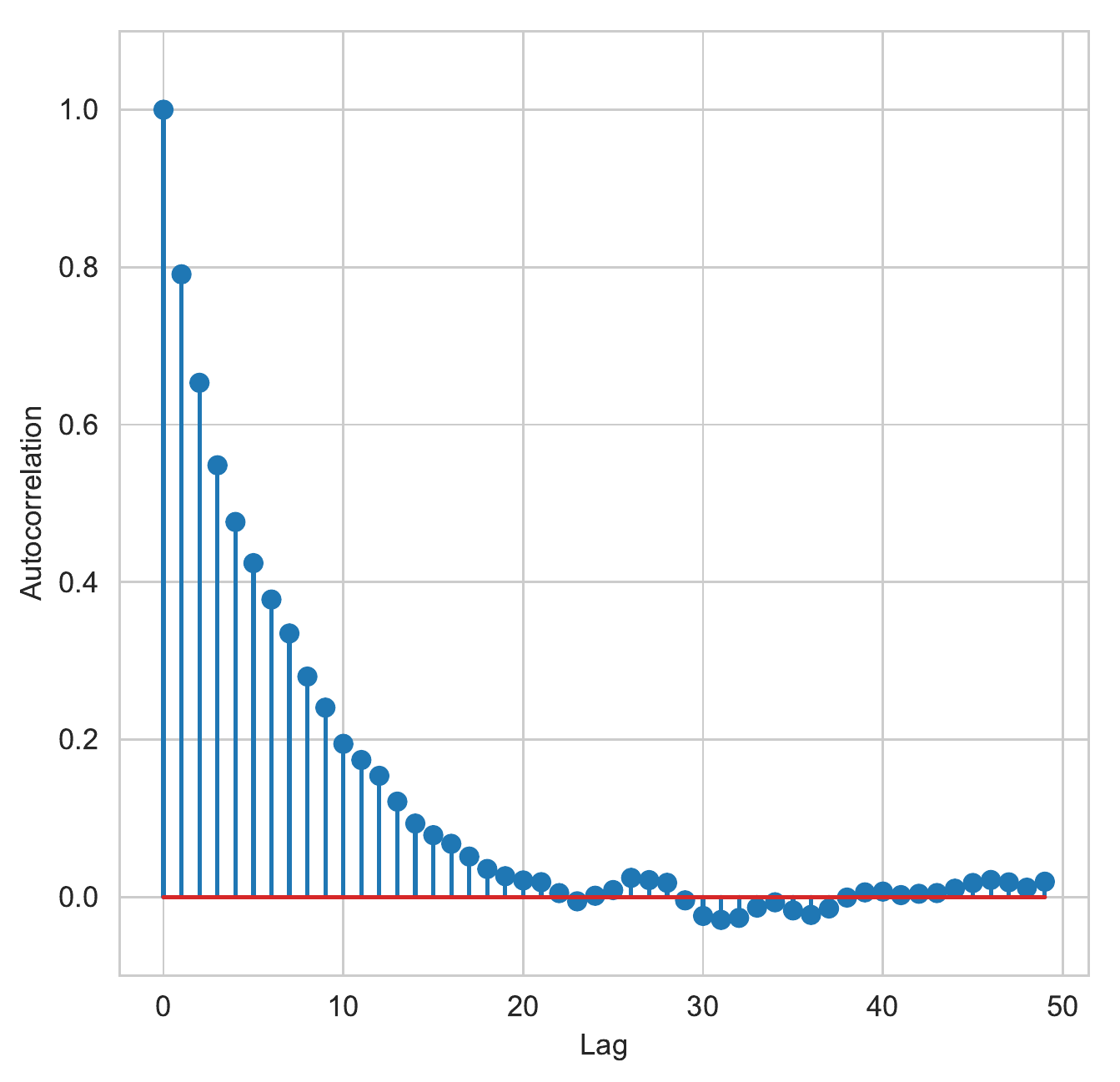}
	\end{center}
	\caption{}
	\label{fig:stoch_vol_autocorr}
\end{subfigure}\vfill
\vspace{-0.5\baselineskip}
\caption{(a) Traceplot of the log-posterior of the stochastic volatility model for all three samplers. (b) Autocorrelation of the energy for the blocked bouncy particle sampler after discarding the first $250$ samples as burn-in. }
\end{figure}

For the blocking strategy, use a spatiotemporal strategy with blocks 9 timepoints wide, 7 spatial dimensions high, and each block has temporal overlap 4 and spatial overlap 3, giving a total of $151 \times 6 = 906$ blocks. Due to the better performance of partitioned blocked bouncy particle sampler in the previous example, we only compare this method with blocked particle Gibbs, see \cite{singh2017blocking}, and the particle Gibbs with ancestor sampling algorithm of \cite{lindsten2014particle}, both using a bootstrap particle filter as proposal mechanism. For the blocked particle Gibbs sampler, we let the blocks be 25 observations wide and have overlap 5. For a fair comparison, we set the number of particles to 500 which leads to an average time per sample quite close to that of the spatiotemporal blocked bouncy particle sampler for both samplers. We generated 2000 samples via each algorithm, and initialized each at the $d \times N$ zero vector, and for the velocity we used the $d \times N$ vector of ones. Typically, estimation of latent states will be carried out inside a Gibbs sampling algorithm that also estimates parameters, indicating that prior knowledge of the states are retained, whereas this example tests the significantly more difficult case of no prior information on the latent states. 

In Figure \ref{fig:stochvol_traceplot}, we illustrate the posterior energy. The blocked particle Gibbs sampler moves in a wide band of posterior energy, but never reaches levels of higher posterior probability. This is in contrast to the results reported in  \cite{singh2017blocking} where the dimension of the hidden state is much lower and thus the state transition density has better forgetting properties than our higher dimensional example. Even if this issue could be remedied, see \cite{bunch2015particle}, implementing particle Gibbs with both temporal and spatial blocking appears non-trivial in contrast to the ease of which it can be achieved with the BPS. The ancestor sampling-based particle Gibbs sampler similarly does not generate proposals that have high posterior probability. Conversely, the bBPS reaches stationarity in less than 100 samples, and subsequently mixes across the posterior: the auto-correlation function, plotted in Figure \ref{fig:stoch_vol_autocorr}, reaches zero around a lag of 20 samples, indicating adequate mixing for a posterior of this dimension. In Figure \ref{corr_estimates}, we plot the correlation matrix of the assets, and also the estimated latent volatility via the posterior mean. It is quite clear that the volatilities show weaker correlation across the assets, but appear to preserve some of the structure of seen in the correlation matrix of the log returns.

\begin{figure}
\begin{center}
\includegraphics[width=\textwidth]{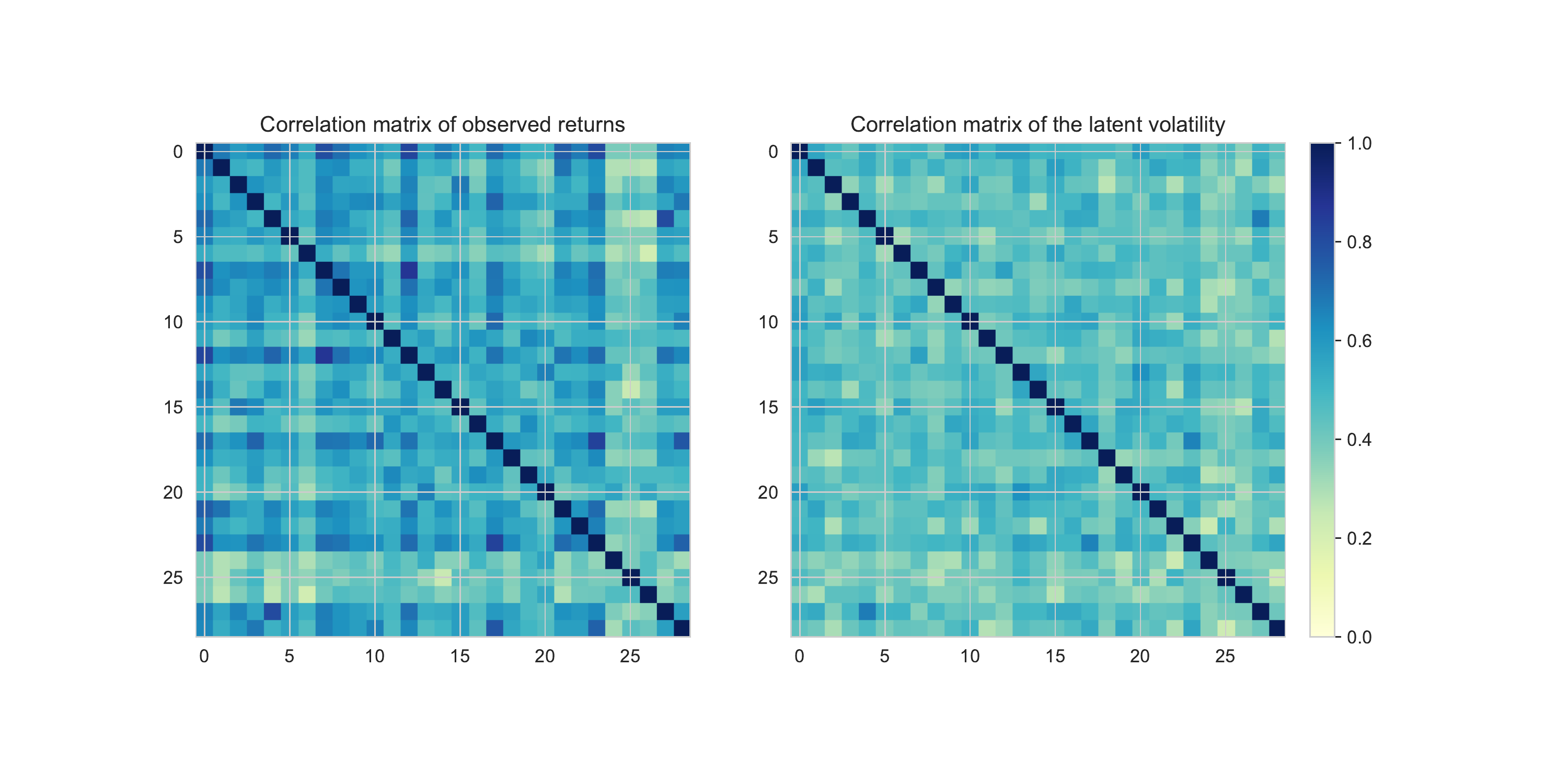}
\end{center}
\caption{Left: estimated correlation matrix from the log-returns over the entire period. Right: estimated correlation matrix of the latent volatilities from the posterior mean estimate from the even-odd bBPS. }\label{corr_estimates}
\end{figure}

\section*{Acknowledgement}
JVG acknowledges financial support from an EPSRC Doctoral Training Award.

\appendix

\section{Proofs}
\label{section:proofs}
%In what follows we ignore super/sub-script 0 for the event-time, position and velocity. We will prove the lemmas in the order that facilitates reading. 
\subsection{Proof of Proposition \ref{prop:bbps}}
\label{proof:bbps}
\begin{proof}
{Denoting the domain of $\mathcal L^{\textsc{bBPS}} $ by $D(\mathcal L^{\textsc{bBPS}})$, invariance follows if $\forall f \in \overline{D} \subset D(\mathcal L^{\textsc{bBPS}})$, with $\overline{D}$ a core for $(\mathcal L^{\textsc{bBPS}}, D(\mathcal L^{\textsc{bBPS}}))$, that
\begin{align*}
\int \mathcal L^{\textsc{bBPS}} f(x,v) \pi(dx)p(dv) = 0,
\end{align*}
see \cite[Chapter 9]{ethier2009markov}. Identification of the core for the generators associated with piece-wise deterministic Markov processes is quite technical, we refer to \cite[Section 7]{durmus2018piecewise} and \cite[Section 5]{holderrieth2019cores} for details. We from now on assume for the test function that $f \in C^{1,0}_0( (\mathbb R^d \times \mathbb R^N) \times (\mathbb{R}^d \times \mathbb{R}^N) \rightarrow \mathbb{R})$, for the density that $U(x) \in C^1(\mathbb{R}^d \times \mathbb{R}^N \rightarrow \mathbb{R})$.}
%, and furthermore that $U(x)$ is bounded.}
The proof in essence follows that of Proposition 1 in \cite{bouchard2018bouncy}.  For the first part of the generator associated with the linear flow, consider first the integral with respect to $\pi(x)$. We have
\begin{eqnarray*}
\int \langle \nabla_x f(x, v), \phi \star v \rangle_F \pi(dx) = \frac 1 Z \int \langle \nabla_x f(x, v),  \phi \star  v \rangle_F e^{-U(x)} dx,
\end{eqnarray*}
where $Z$ is the normalising constant $Z = \int e^{-U(x)}$. Applying integration-by-parts we immediately get
\begin{eqnarray*}
\frac 1 Z \int \langle \nabla_x f(x, v),  \phi \star v \rangle_F e^{-U(x)} dx &=&  %\frac 1 Z \langle f(x,v),  \phi \star v \rangle_F e^{-U(x)} \Big |_{-\infty}^{\infty} + 
\frac 1 Z \int  f(x,v) \langle \nabla U(x),  \phi \star v\rangle_F e^{-U(x)}dx \\
&=& \int f(x,v) \langle \nabla U(x), \phi \star v \rangle_F \pi(dx)
\end{eqnarray*}
by integrability of the functions in the domain of $\mathcal L^{bBPS}$. For the second part, note that
\begin{eqnarray*}
	\sum_{B \in \overline B} \lambda_B(x, \textsc{reflect}_x^B (v)) - \lambda_B(x,v) &=& \sum_{B \in \overline B} \big \{-\langle \nabla_B U(x), v_B \rangle_F \big \}^+ -  \big \{\langle \nabla_B U(x), v_B\rangle_F \big \}^+ \\
	&=& \sum_{B \in \overline B} - \langle \nabla_B U(x), v_B \rangle_F  \\
	%&=& -\sum_{B \in \overline B} \sum_{\substack{(k,n) \in B \\ (k,n) \notin B', \forall B' \in \overline B}} \nabla^B U_n^k (x) \cdot v_n^k \\
	%&+& \sum_{k=1}^d \sum_{n=1}^N \mathds{1}_{\phi_n^k> 1}(i) \phi_n^k \nabla U^k_n(x) \cdot v_n^k \\\
	&=& -\langle \nabla U(x), \phi \star v \rangle_F.
\end{eqnarray*}
where the last line follows by the definition of $\phi$.
Consider then the integral of the jump dynamics generator
\[
 \int \int\sum_{B \in \overline B} \lambda_B(x,v)\Big [f(x, \textsc{reflect}_x^B (v)) - f(x,v) \Big ] \pi(dx) p(dv).
\]
Using that $(\textsc{reflect}^{B}_x)^{-1} = \textsc{reflect}^{B}_x$ by involution and the norm-preserving property of the restricted reflection operator we get
\[
\int  \lambda_B(x,v)f(x, \textsc{reflect}_x^B (v))\pi(dx) p(dv) = \int \lambda_B(x,\textsc{reflect}_x^B(v))f(x,  v) \pi(dx) p(dv),
\]
so we have from using the identity above that
\begin{eqnarray*}
 \int \int \sum_{B \in \overline B} \lambda_B(x,v)\Big [f(x, \textsc{reflect}_x^B (v)) - f(x,v) \Big ] \pi(dx) p(dv) \\ = \int \int \sum_{B \in \overline B} \Big [\lambda_B(x, \textsc{reflect}_x^B (v) - \lambda_B(x,v) \Big ]f(x,v) \pi(dx) p(dv) 
 = - \int \int\langle \nabla U(x), \phi \star v \rangle_F f(x,v) \pi(dx) p(dv),
\end{eqnarray*}
which implies the result.
\end{proof}

\subsection{Proof of Proposition \ref{prop:pibps_invariance}}
\label{proof:pibps_invariance}
\begin{proof}
We will show that the eoBPS is a special case of the blocked bouncy particle sampler, we again subdue dependence on refreshments. Writing out the integral with respect to $Q^B_x,$ we have
\begin{align*}
    \mathcal L^{\textsc{eoBPS}}f(x,v) &=  \langle \nabla_x f(x,v), \phi \star v \rangle_F \\ &+ \sum_{\kappa \in \{odd, even \}}  \widehat \Lambda^\kappa(x,v) 
    \left [ 
    \sum_{B \in \overline B_\kappa} 
        \left [
        \frac{\lambda_B(x,v)}{\widehat\Lambda^\kappa(x,v)}f(x, \textsc{reflect}_x^B (v)) + \left (1-\frac{\lambda_B(x,v)}{\widehat\Lambda^\kappa(x,v)} \right) f(x,v) - f(x,v)
        \right  ]
    \right ] \\
    &=\langle \nabla_x f(x,v), \phi \star v \rangle +  
    \sum_{\kappa \in \{odd, even \}}  \widehat \Lambda^\kappa(x,v) 
    \left [ 
        \sum_{B \in \overline B_\kappa} 
        \left [
        \frac{\lambda_B(x,v)}{\widehat\Lambda^\kappa(x,v)}f(x, \textsc{reflect}_x^B (v))  -  \frac{\lambda_B(x,v)}{\widehat\Lambda^\kappa(x,v)} f(x,v)
        \right  ]
    \right ] \\
    &=\langle \nabla_x f(x,v), \phi \star v \rangle +  
    \sum_{\kappa \in \{odd, even \}}     
    \left [ 
        \sum_{B \in \overline B_\kappa} 
        \left [
        \lambda_B(x,v) f(x, \textsc{reflect}_x^B (v))  - \lambda_B(x,v) f(x,v)
        \right  ]
    \right ]  \\
    &= \langle \nabla_x f(x,v), \phi \star v \rangle +
    \sum_{B \in \overline B}  \lambda_B(x,v)
    \left [
         f(x, \textsc{reflect}_x^B (v))  - f(x,v)
    \right ] \\
    &= \mathcal L^{\textsc{bBPS}}f(x,v),
\end{align*}
which by Proposition \ref{prop:bbps} gives the result.
\end{proof}
\subsection{Proof of Lemma \ref{lemma:bound_rate}}
\label{proof:event_rates}
\begin{proof}
We will just consider the odd strategy in the proof, everything translates seamlessly. By Hölders inequality, we have for any $p \in  \mathbb N$

\begin{align*}
\mathbb E_{\pi \otimes p_v} (\widehat{\Lambda}_{odd}) &= 
\mathbb E_{\pi \otimes p_v} (\max_{B \in \overline B_{odd}} \lambda_B)
\\ &
\leq \left ( \mathbb E_{\pi \otimes p_v} \max_{B \in \overline B_{odd}}  | \lambda_B|^p \right )^{\frac{1}{p}} \\ 
 &
\leq \left ( \mathbb E_{\pi \otimes p_v} \sum_{B \in \overline B_{odd}}  | \lambda_B|^p \right )^{\frac{1}{p}} 
 \\
&\leq|\overline B_{odd}|^{\frac{1}{p}} \max_{B \in \overline B_{odd}} \left (\mathbb E_{\pi \otimes p_v} | \lambda_B |^p \right )^{\frac{1}{p}}
\end{align*}
For any $B \in \overline{B}_{odd},$ by adapting the proof of Lemma 1.10 in \cite{rigollet2015high}, we have by the sure positivity of the rate function that
\begin{align*}
    \mathbb E_{\pi \otimes p_v} | \lambda_B |^p &= \int_{[0,\infty)} \mathbb P(\lambda_B \geq t^{1/p}) dt \\
    &\leq 2\int_{[0,\infty)} e^{-2\alpha t^{1/p}}dt \\
    &= \frac{2p}{(2\alpha)^p} \int_{[0,\infty)} e^{-u} u^{p-1} du  \\
    &= \frac{1}{\alpha^p} p!,
\end{align*}
such that in particular
\begin{align}
    \mathbb E_{\pi \otimes p_v} \widehat{\Lambda}_{odd} \leq 
    |\overline B_{odd}|^{\frac{1}{p}} \frac{2p}{\alpha}.
\end{align}
Optimizing over $p$ leads to $p^* = \log |\overline B_{odd}| $, which together with the bound gives
\begin{align*}
    \mathbb E_{\pi \otimes p_v} \widehat{\Lambda}_{odd} \leq |\overline B_{odd}|^{\frac{1}{\log |\overline{B}_{odd}|}} \frac{2\log |\overline{B}_{odd}|}{\alpha} \leq \frac{2e}{\alpha} \log |\overline{B}_{odd}|,
\end{align*}
implying the result. 
\end{proof}

\section{Parameters of the stochastic volatility model}
\subsection{Choice of parameters}
The daily asset prices $(p_n)_{n=1,2,\ldots, N}$ is collected from eSignal Data Access and transformed to log returns via the relation $y_n = \log p_n / p_{n-1}$. As our paper is centered on latent state estimation, we have foregone a full Bayesian parameter estimation. Instead, for all unobserved quantities we have used the parameters proposed by \cite{ishihara2012efficient} Section 3, these are based on previous empirical studies by the authors and quite closely correspond to what is inferred during their parameter estimation procedure on S\&P 500 data. In particular, we set the persistence parameter to $\alpha = 0.99$ and use $\nu = 15$ degrees of freedom for the multivariate t-distribution. 

For the unobserved volatility covariance matrix $\Sigma_\eta$, the cross-asset correlation is set at $0.7$, and the same standard deviation is assumed for each asset, $0.2$. 
For the leverage matrix, the intra-asset parameter is set at $\Sigma_{\rho, ii} = -0.4$ and cross-asset leverage $\Sigma_{\rho,ij} = -0.3$.

We estimate the return covariance matrix $\Sigma_\epsilon$ directly from the observed log returns over the entire period. The values we arrive at from this procedure is again close to what is empirically observed, indicating it is a reasonable parameter value to use for a latent states estimate. 
If we were to run a spatially blocked scheme, we could for example apply a hierarchical clustering algorithm like the algorithm of \cite{ward1963hierarchical} to learn the relationship between the assets, and then rearrange the order of the assets to match the order of the clustering procedure. As discussed this should have a beneficial effect on mixing, as the blocks become more localised.

\bibliography{bibfile}

\end{document}